\newif\ifIEEE
\newif\ifACM

\documentclass[10pt, conference, compsocconf]{IEEEtran}
\IEEEtrue

\usepackage{graphicx,amsmath,amsfonts}
\usepackage[english]{babel}
\usepackage[utf8]{inputenc}
\usepackage[T1]{fontenc}
 \usepackage{pifont}
\usepackage{rotating}
\usepackage{color}
 \usepackage{bbm}
 \usepackage{wasysym}
   \usepackage{stmaryrd}
   \usepackage{multicol}
   
\usepackage{algorithm}
\usepackage{algorithmicx}
\usepackage[noend]{algpseudocode} 

\usepackage{amsthm}
\usepackage{mathtools}
\usepackage{tensor}
\usepackage{interval}
\usepackage{fancyhdr}
\usepackage{datetime}
\usepackage{enumerate}

\ifIEEE
\usepackage{lastpage}
\fi

\ifIEEE
\newlength\shadedboxwidth

\fi

\author{Grzegorz Głuch, Jerzy Marcinkowski, Piotr Ostropolski-Nalewaja\\
Institute of Computer Science, University of Wrocław} 

\title{Can One Escape Red Chains?\\ Regular Path Queries Determinacy is Undecidable.}
\date{January 2018}

\newtheorem{theorem}{Theorem}

\ifIEEE
\newtheorem{lemma}[theorem]{Lemma}
\newtheorem{definition}[theorem]{Definition}
\fi

\numberwithin{theorem}{section}

\newtheorem{observation}[theorem]{Observation}

\newtheorem{notation}[theorem]{Notation}
\newtheorem{exercise}[theorem]{Exercise}

\begin{document}

\maketitle
\newcommand{\DOM}[1]{\textit{Dom}(\mathbb{#1})}
\newcommand{\ddd}[0]{\ldots}
\newcommand{\rtgd}[2]{#1 \rightarrow #2}
\newcommand{\rcrg}[1]{\rc{#1}{R}{G}}
\newcommand{\rcgr}[1]{\rc{#1}{G}{R}}
\newcommand{\rc}[3]{\rtgd{#2(#1)}{#3(#1)}}

\newcommand{\step}[0]{step}

\newcommand{\Path}[2]{#1(#2)}
\newcommand{\barepath}[1]{\mathcal{P}_{#1}}
\newcommand{\requests}[0]{rq}
\newcommand{\database}[0]{\mathbb{D}}
\newcommand{\history}[0]{\mathbb{H}}
\newcommand{\histories}[0]{\Omega}
\newcommand{\pair}[1]{\langle #1 \rangle}
\newcommand{\set}[1]{\{#1\}}
\newcommand{\naturals}{\mathbb{N}}
\newcommand{\chase}[1]{\textit{Chase}(#1)}

\newcommand{\CC}[5]{(\tensor*[^{#5}_{#4}]{\mathbf{#1}}{^{#3}_{#2}})}
\newcommand{\ver}[0]{V}
\newcommand{\hor}[0]{H}

\newcommand{\layer}{\mathbb{L}}

\ifACM
\DeclareRobustCommand*\cal{\@fontswitch\relax\mathcal}
\fi

\newcommand{\outline}[1]{
\noindent
\underline{~~~~~~~~~~~~~~~~~~~~~~~~~~~~~~~~~~~~~~~~~~~~~~~~~~~~~~~~~~~~~~~~~~~~~~~}\vspace{-1mm}\\
\noindent
\underline{~~~~~~~~~~~~~~~~~~~~~~~~~~~~~~{\sc outline}~~~~~~~~~~~~~~~~~~~~~~~~~~~~~~}\\

\vspace{-2mm} #1 \vspace{-2mm}

\noindent
\underline{~~~~~~~~~~~~~~~~~~~~~~~~~~~~~~~~~~~~~~~~~~~~~~~~~~~~~~~~~~~~~~~~~~~~~~~}\vspace{-5mm}\\
\noindent
\underline{~~~~~~~~~~~~~~~~~~~~~~~~~~~~~~~~~~~~~~~~~~~~~~~~~~~~~~~~~~~~~~~~~~~~~~~}\\
}
  
  \vspace{-1.5mm}
  \noindent
{\bf Abstract.}
For a given set of queries (which are expressions in some query language) $\mathcal{Q}=\{Q_1$, $Q_2, \ldots Q_k\}$ and for another query $Q_0$ we say that 
 $\mathcal{Q}$ determines $Q_0$ if -- informally speaking --  for every database $\mathbb D$, the information 
 contained in the views $\mathcal{Q}({\mathbb D})$ is sufficient to compute $Q_0({\mathbb D})$.
 
 Query Determinacy Problem is the problem of deciding, for given  $\mathcal{Q}$ and $Q_0$, whether $\mathcal{Q}$ determines $Q_0$.
 Many versions of this problem, for different query languages, were studied in database theory. In this paper we 
 solve a problem stated in [CGLV02] 
 and show that 
 Query Determinacy Problem is undecidable for the Regular Path Queries -- the paradigmatic query language of graph databases.
 
\vspace{-0.5mm}
 
\section{Introduction}

\noindent
{\bf Query determinacy problem (QDP).} Imagine there is a database $\database$
we have no direct access to, and there are views of this $\database$ available to us, defined by some set of
 queries $\mathcal{Q}=\{Q_1$, $Q_2, \ldots Q_k\}$ (where the language of queries from $\mathcal{Q}$ is a parameter of the problem).
 And we are given another  query $Q_0$. Will we be able, regardless of $\mathbb D$,
 to compute $Q_0(\mathbb D)$
only using the views   $Q_1(\mathbb D), Q_2(\mathbb D),\ldots Q_k(\mathbb D) $? The answer depends on whether the queries in $\mathcal {Q}$ 
{\em determine\footnote{
Or, using the language of [CGLV00], [CGLV00a] [CGLV02] and [CGLV02a], whether $\mathcal{Q}$ are lossless with respect to $Q_0$.  }} query $Q_0$. 
Stating it more precisely, the  {\bf Query Determinacy Problem   is}\footnote{More precisely, the problem
  comes in two different flavors, ``finite'' and ``unrestricted'', depending on whether the ($\clubsuit$) ``each'' ranges over finite structures only, or all structures, including infinite.}:
\vspace{1.5mm}

\noindent\fbox{%
    \parbox{\linewidth}{%
  The instance of the problem is a set of  queries $\mathcal{Q}=\{Q_1,\ldots Q_k\}$, and 
another  query $Q_0$.

 The question is whether $\mathcal{Q}$ determines $Q_0$, which means that for ($\clubsuit$) each
 two structures (database instances) ${\mathbb D}_1$ and ${\mathbb D}_2$ such that  
$Q({\mathbb D}_1)= Q({\mathbb D}_2)$ for each $Q\in \mathcal{Q}$, it also holds that  $Q_0({\mathbb D}_1) = Q_0({\mathbb D}_2)$.}%
}
\vspace{1.5mm}

\noindent
QDP is seen as a very natural problem in the area of database theory, with a 30 years long history as a research subject   --
see for example 
[H01], or Nadime Francis thesis [F15] for a survey. In [DPT99] QDP naturally appears in the context of query evaluation plans optimization. 
More recent examples are [FG12], where the context for QDP is the view update problem or 
 [FKN13], where the context is description logics.  In the above examples
 the goal  is  optimization/efficiency so 
 we ``prefer'' $Q_0$ to be determined by $\mathcal{Q}$. 
 Another context, where it is ``preferred'' that $Q_0$ is not determined, 
 is privacy: we would like to release some views of the database, but in a way that does not allow certain query to be computed.

The oldest paper we were able to trace, where QDP is studied, is [LY85].
Over the next 30 years 
many decidable and undecidable cases have been identified. Let us
just cite some more recent results: [NSV10] shows that the problem is decidable for conjunctive queries 
if each  query from $\mathcal{Q}$ has only one free variable; in 
[A11] decidability is shown for $\mathcal{Q}$ and $Q_0$ being ''conjunctive path queries''. 
This is generalized in [P11]  to the the scenario where $\mathcal{Q}$ are conjunctive path queries but $Q_0$ is any conjunctive
query. 

The paper  [NSV06] was the first to 
present a negative result. QDP was shown there to be undecidable if unions of conjunctive queries are allowed in  $\mathcal{Q}$ and $Q_0$. 
 In [NSV10] it was proved that determinacy is also undecidable if the elements of $\mathcal{Q}$ are conjunctive queries  and $Q_0$ is a first order sentence 
(or the other way round). Another  negative result is presented in [FGZ12]: determinacy is shown there to be undecidable  if  $\mathcal{Q}$  is a DATALOG program and 
$Q_0$ is a conjunctive query. Finally,  closing the classification for the traditional relational model, it was shown in [GM15] and [GM16] that QDP
is undecidable for $Q_0$ and the queries in $\mathcal{Q}$ being conjunctive queries.\\[1mm]
{\bf QDP for Regular Path Queries.} 
While the determinacy problem is now well understood for the pure relational model\footnote{Apparently, when talking about the 
relational model, there may  still be some work to do concerning QDP in the context of bag semantics, see [GB14].}, 
it has been, for a long time, open  for the  graph databases scenario.
In 
this scenario, the underlying data is  modeled as graphs, in which nodes
are objects, and edge labels define relationships between those objects.
Querying  such graph-structured data has received much attention recently, due
to numerous applications, especially for the social networks.

 There are many more or less expressive query languages for  such databases (see [B13]). The core of all of them (the SQL of graph databases) is
  RPQ -- the language of Regular Path Queries. RPQ queries
ask for all pairs of objects in the database that are connected by a specified
path, where the natural choice of the path specification language, as [V16] elegantly explains, is the language of regular expressions.
This idea is at least 30 years old (see for example [CMW87, CM90]) and considerable effort was put to
create tools for  reasoning about regular path queries, analogous to the ones we have in the 
 traditional relational databases  context.
For example [AV97] and  [BFW98] investigate  decidability of the
implication problem for path constraints, which are integrity constraints used for  RPQ optimization. Also, containment of
conjunctions of regular path queries has been addressed and proved decidable in
[CDGL98] and  [FLS98], and then, in more general setting, in [JV09] and [RRV15]

It is natural that also query determinacy problem has been stated, and studied,
for Regular Path Queries model. This line of
research was initiated in 
[CGLV00], [CGLV00a] \hfill \break [CGLV02] and [CGLV02a], 
and it was [CGLV02] where the central problem of this area -- decidability of QDP for RPQ was first stated (called there ``losslessness for exact semantics'')

A method for
computing a rewriting of a regular path query  in terms of other regular
expressions (if such rewriting exists) \footnote{existence of rewriting is a related property to determinacy, but stronger} 
is shown in [CGLV02]. And it is proven  that it is 2ExpSpace-complete to decide whether
there exists a rewriting of the query  that can be expressed as a regular
path query. Then a notion of monotone determinacy is defined, meaning that not only  $Q_0({\mathbb D})$ is a function\footnote {$\mathbb D$ is an argument here. 
Saying that ``$Q_0(\mathbb{D})$ is a function $\mathcal{Q}({\mathbb D})$'' is equivalent to saying that ${\mathcal Q}$ determines $Q_0$.} of ${\mathcal Q}({\mathbb D})$ but this function is also monotone --
the greater  ${\mathcal Q}({\mathbb D})$ (in the inclusion ordering) the greater $Q_0({\mathbb D})$, and it is shown that monotone determinacy is decidable 
in ExpSpace. This  proves that monotone
determinacy,
which is -- like rewritability -- also a notion related to determinacy but stronger, 
does not coincide with the existence of a regular path rewriting, which is  2ExpSpace-complete (while of course the existence of rewriting implies monotonicity).
This proof is indirect and it is interesting that a specific example separating monotone determinacy and rewritability has only been shown in 
[FSS14].
However, [CGLV02a] also provides an example where a regular
path view determines a regular path query in a non-monotone way  showing  that, in this setting, determinacy does not coincide with
monotone determinacy.

 In [CGLV02], apart from the standard QDP, the authors consider the so called ``losslessness under sound semantics''. They show that computing ``certain answers'' (under this semantics) 
 of a regular path query with
respect to a regular path view reduces to the satisfiability of (the negation of)  uniform
CSP (constraint satisfaction problem). Building on this connection and on the known links between CSP and Datalog [FV98],
they show how to compute approximations of this CSP in Datalog.
This is studied in more detail in [FSS14] and a surprising result is proved, that 
when a regular path view determines a regular path query in a monotone way, then one of the approximations is exact.

But, despite the considerable body of work in the area around the the main problem,  little was so far known about the problem of decidability of QDP for RPQ itself.
On the positive side, the previously mentioned result of Afrati [A11] can be seen as a special case, where each of the regular 
languages (defining the queries) only consists of one word (path queries, considered in [A11] constitute in fact the intersection of CQ and RPQ). 
Another positive result is presented in [F17], where ``approximate determinacy'' is shown to be decidable 
if the query $Q_0$ is (defined by) a single-word regular language, and the languages defining the queries in $Q_0$ and  $\mathcal Q$ are  over a single-letter alphabet.
The failure to solve the problem completely even for this very simple variant shows how complicated things very quickly become. But it is the analysis which is so obviously hard (not QDP itself as a computational problem) and  
it is not immediately clear how QDP for RPQ  could be used to encode anything within. In consequence, no lower bounds have been  known so far, except of 
a simple one from [F15], where undecidability is shown if $Q_0$ can be context-free rather than just regular.

\vspace{2mm}
\noindent
{\bf Our contribution.} 
The main result of this paper is:

\vspace{-1mm}
\begin{theorem}\label{i1}
QDP-RPQ, the Query Determinacy Problem for Regular Path Queries, is undecidable.
\end{theorem}

\vspace{-1mm}
To be more precise, we show that the problem, both in the ``finite'' and 
the ``unrestricted'' version, is co-r.e.-hard, which means that if we take, as an input to our encoding, a Turing machine which accepts (the empty input) then, as the result of the encoding
we get a negative instance of QDP (``no determinacy''), and if we begin from a non-accepting machine then the resulting instance is positive. Notice that this gives the precise 
bound on the complexity of the ``finite'' version of QDP for RPQ -- it is easy to see that finite non-determinacy is recursively enumerable. But there is no such upper bound for the ``unrestricted'' case, and
we are not sure what the precise complexity can be. We believe that the problem may be harder than co-r.e.-complete. 

Regarding the technique we use: clearly we were tempted to save as much as possible from the techniques of [GM15] and [GM16]. 
But hardly anything survived in the new situation (one exception is that the idea of the green-red Chase from [G15]  evolved into the notion of Escape here).
The two important constructions in [GM15] and [GM16] used queries with high number of free variables (this is where states of the Turing machine are encoded, in the form of 
spiders with fancy colorings) and queries which can be homomorphically, non-trivially, mapped into themselves -- this is how the original small structure (``green spider'' in 
[GM15] and [GM16] or (green) ${\mathbb D}_0$ in this paper) could grow. None of the mechanisms is available in the current context, so in principle the whole proof was built from scratch.




\vspace{2mm}
\noindent
{\bf Remark.} [B13] makes a distinction between ``simple paths semantics'' for Recursive Path Queries and ``all paths semantics''. As all the graphs we produce in this paper are acyclic (DAGs), 
all our results hold for both semantics.

\vspace{2mm}
\noindent
{\bf Organization of the paper}
The rest of this paper is devoted to the proof of Theorem \ref{i1}. In short Section \ref{preliminaries} we introduce the (very few) notions and some notations we need to use. 

In Section \ref{sygnatura} we first follow the ideas from [GM15] defining red-green signature. Then we 
define the game of Escape and state a crucial lemma (Lemma \ref{bridge}), asserting that this game really fully characterizes determinacy for Recursive Path Queries. In Section \ref{uniwersalnosc} we prove this Lemma.

At this point we will have all the tools ready for proving Theorem \ref{i1}. In Section \ref{redukcja} we
explain what is the undecidable problem we use for our reduction, and present the reduction.
In Sections \ref{przewodnik} -- \ref{ulga} we use the characterization provided by Lemma \ref{bridge} to prove correctness of this reduction.


\section{Preliminaries}\label{preliminaries}


\noindent{\bf Structures.}
When we say ``structure" we always mean a directed graph with edges labeled with letters from some signature/alphabet $\Sigma$. In other words every structure we consider is relational structure $\database$ over some signature $\Sigma$ consisting of binary predicate names. Letters $\mathbb{D}$, $\mathbb{M}$, $\mathbb{G}$ and $\mathbb{H}$ are used to denote structures. $\Omega$ is used for a set of structures.


For two structures ${\mathbb G}$ and ${\mathbb G'}$ over $\Sigma$, with sets of vertices $V$ and $V'$, a function $h:V \rightarrow V'$ is (as always) called a 
homomorphism if for each two vertices $\pair{x,y}$ connected by an edge with label $E\in \Sigma$ in $\mathbb{G}$ there is an edge connecting $\pair{h(x),h(y)}$, with the same label $E$, in $\mathbb{G'}$.

\vspace{2mm}
\noindent{\bf Chains and chain queries.}
Given a set of binary predicate names $\Sigma$ and a word $w = a_1a_2\ddd a_n$ over $\Sigma^*$ we 
define a chain query $ w(x_0, x_n)$ as a conjunctive query:

\begin{center}
$\exists_{x_1,\ddd,x_{n-1}} a_1(x_0, x_1) \wedge  a_2(x_1, x_2) \wedge \ddd a_n(x_{n-1}, x_n).$
\end{center}

We use the notation $w[x_0,x_n]$ to denote the canonical structure (``frozen body'') of query $w(x_0, x_n)$ --
the structure consisting of elements $x_0,x_1,\ldots x_n$ and atoms $a_1(x_0, x_1), \\ a_2(x_1, x_2),\ldots$ $a_n(x_{n-1}, x_n)$.

\vspace{2mm}
\noindent{\bf Regular path queries.}
For a regular language $Q$ over $\Sigma$ we define a query, which is also denoted by $Q$, as:

\begin{center}
$Q(x, y) = \exists_{w \in Q} w(x, y)$
\end{center}

In other words such a query $Q$ looks for a path in the given graph labeled with any word from $Q$ and returns the endpoints of that path.

We use letters $Q$ and $L$ to denote regular languages and $\mathcal{Q}$ and $\mathcal{L}$ to denote sets of regular languages. 
The notation $Q(\database)$ has the natural meaning of: $Q(\database) = \set{\pair{x,y}\, |\, \database \models Q(x,y)}$.



\section{Red-Green Structures and Escape}\label{sygnatura}

\subsection{Red-green signature and Regular Constraints}

For a given alphabet (signature) $\Sigma$ let $\Sigma_G$ and $\Sigma_R$ be two copies of $\Sigma$ one written with "green ink" and another with "red ink". Let $\bar\Sigma = \Sigma_G \cup \Sigma_R$.

For any word $w$ from $\Sigma^*$ let $G(w)$ and $R(w)$ be copies of this word written in green and red respectively. For a regular language $L$ over $\Sigma$ let $G(L)$ and $R(L)$ be copies of this same regular language but over $\Sigma_G$ and $\Sigma_R$ respectively.
Also for any structure $\database$ over $\Sigma$ let $G(\database)$ and $R(\database)$ be copies of this same structure $\database$ but with labels of edges recolored to green and red respectively.

For a pair of regular languages $L$ over $\Sigma$ and $L'$ over $\Sigma'$ we define \textit{Regular Constraint} $\rtgd{L}{L'}$ as a formula

\begin{center}
$\forall_{x,y} L(x,y) \Rightarrow L'(x,y).$
\end{center}

We use the notation $\database \models r$ to say that an RC $r$ is satisfied in $\database$. Also, we write $\database \models T$ for a set $T$ of RCs when for each $t\in T$ it is true that $\database \models t$.

For a graph $\database$ and an RC $t=\rtgd{L}{L'}$ let $\requests(t, \database)$ (as ``requests'') be the set of all triples $\pair{x,y,\rtgd{L}{L'}}$ such that $\database \models L(x,y)$ and $\database \not\models L'(x,y)$. For a set $T$ of RCs by 
$\requests(T, \database)$ we mean the union of all sets $\requests(t, \database)$ such that 
$t \in T$. Requests are there in order to be satisfied:

\vspace{-2mm}
\begin{algorithm}[H]
\begin{algorithmic}[1]
\Statex \textbf{function} \textsc{Add}
\Statex \textbf{arguments}:
\Statex \begin{itemize}
			\item Structure $\database$
            \item RC $\rtgd{L}{L'}$
			\item pair $\pair{x,y}$ such that $\pair{x,y,\rtgd{L}{L'}} \in \requests(\rtgd{L}{L'}, \database)$
		\end{itemize}
\Statex \textbf{body}:
\State Take a word $w = a_0 a_1 \ddd a_n$ from $L'$ and create a new path $w[x,y] = a_0(x, x_1),a_1(x_1,x_2),\ddd,a_n(x_{n-1},y)$ where $x_1,x_2,\ddd,x_{n-1}$ are {\bf new} vertices
\State \Return $\database \cup w[x,y]$.
\end{algorithmic}
\label{alg:local_improvements}
\end{algorithm}

\vspace{-2mm}
Notice that the result $Add(D, \rtgd{L}{L'}, \pair{x,y})$ depends on the choice of $w\in L'$. So the procedure is non-deterministic.

For a  regular language ${L}$ we define 
${L}^\rightarrow = \rcgr{L}$ and ${L}^\leftarrow = \rcrg{L}$. 
All regular constraints we are going to consider are either ${L}^\rightarrow$ or ${L}^\leftarrow$ for some regular $L$. 

For a  regular language $L$ we define
$L^\leftrightarrow=\{L^\rightarrow, L^\leftarrow\}$ and for a set $\mathcal{L}$ of regular languages we define:

$${\mathcal L}^\leftrightarrow = \bigcup_{L \in \mathcal{L}} L^\leftrightarrow.$$

Requests of the form $\langle x,y,t\rangle$ for some RC $t\in L^\rightarrow$
($t\in L^\leftarrow$)
are 
{\em generated by} $G(L)$ (resp. {\em by}  $R(L))$. Both groups jointly are said to be 
{\em generated by} $L$. 

The following lemma is straightforward to prove and characterizes determinacy in terms of regular constraints:

\begin{lemma}
\label{lm-det-struct}
A set $\mathcal{Q}$ of regular path queries over $\Sigma$ does not determine regular path query  $Q_0$, over the same alphabet,  if and only if there exists a structure $\mathbb M$ and a pair of vertices $a,b \in {\mathbb M}$ such that ${\mathbb M} \models \mathcal{Q}^\leftrightarrow$ and ${\mathbb M} \models {(G(Q_0))}(a,b)$ but   ${\mathbb M}\not\models {(R(Q_0))}(a,b)$.
\end{lemma}

Any structure ${\mathbb M}$, as above, will be called \textit{counterexample}.

%

\subsection{The game of Escape}\label{game}
An instance Escape($Q_0$, $\mathcal{Q}$) of a solitary game called \textit{Escape}, played by a player called \textit{Fugitive}, is:
\begin{itemize}
\item a regular language $Q_0$ of {\em forbidden chains} over $\Sigma$.
\item a set of regular languages $\mathcal{Q}$ over $\Sigma$,
\end{itemize}

The rules of the game are:
\begin{itemize}
	\item First Fugitive picks the \textit{initial position} of the game as $\database_0 = (G(w))[a,b]$ for some $w \in Q_0$.

\item Suppose $\database_i$ is the position of the game after Fugitive move $i$ and $S_i=\requests(\mathcal{Q}^\leftrightarrow, \database_i)$. Then, in move $i+1$, Fugitive can move to any position of the form:
$$\database_{i+1} = \bigcup_{\pair{x,y,t} \in S_i} Add(\database_i, t, \pair{x,y})$$
    
    \item Fugitive loses when for a \textit{final position} $\history = \bigcup\limits_{i=0}^{\infty} \database_{i}$ it is true that $\history \models (R(Q_0))(a,b)$.
\end{itemize}

Let us note that $\database_{i+1} = \database_i$ when $\requests(\mathcal{Q}^\leftrightarrow, \database_i)$ is empty. 

It also would not hurt if, before proceeding with the reading, the Reader wanted to solve:

\begin{exercise}
Notice that if $i$ is even (odd) then  all the requests from ${\mathcal S}_i$ are generated by  $G(L)$ (resp. R(L)), for some $L\in \mathcal{Q}$  which means that all the edges added by {\em Fugitive} in his move $i+1$ are red (resp. green).
\end{exercise}

Let $\step$ be ternary relation such that $\pair{\database,\database',\mathcal{L}} \in \step$  when $\database'$ can be the result of one move of Fugitive, in position  $\database$, in the game of Escape with set of regular languages $\mathcal{L}$.

Obviously, different strategies of Fugitive may lead to different final positions. We will denote set of all final positions reachable from a starting structure $\database_0$, for a set of regular languages $\mathcal{L}$, as $\histories(\mathcal{L}^\leftrightarrow, \database_0)$.

Now we can state the crucial Lemma, that connects the game of Escape and QDP-RPQ:

\begin{lemma}\label{bridge}
For an instance of QDP-RPQ consisting of regular language $Q_0$ over $\Sigma$ and a set of regular languages $\mathcal{Q}$ over $\Sigma$ the two conditions are equivalent:
\begin{enumerate}[(i)]
\item $\mathcal{Q}$ does not determine $Q_0$
\item Fugitive has a winning strategy in Escape($Q_0$, $\mathcal{Q}$).
\end{enumerate}
\end{lemma}


\subsection{Universality of Escape. Proof of Lemma~\ref{bridge}}\label{uniwersalnosc}

\noindent
First  let us leave it as an easy exercise for the Reader to prove:

\begin {lemma}\label{rcsatisfied}
For each set of RCs $T$, for each initial position $ \database_0$ and for 
 each ${\history \in \histories(T,\database_0)}$ it holds that  $\history \models T$.
\end{lemma}

With the above Lemma, the proof of Lemma \ref{bridge} (ii)$\Rightarrow$(i) 
 is  straightforward:
the winning final position of Fugitive can serve as the counterexample $\mathbb M$ from Lemma~\ref{lm-det-struct}.

The opposite direction, (i)$\Rightarrow$(ii) is not completely obvious. Notice that it could {\em a priori} happen that, while
some counterexample exists, it is some terribly complicated structure which cannot be constructed as a final position in a play of the game of Escape. 
We should mention here that all the notions of Section \ref{sygnatura} have their counterparts in [G15]. Instead of Regular Constrains however, in [G15] one finds conventional Tuple Generating Dependencies\footnote{Notice that if all each of the languages in $\mathcal Q$ consists of a single word, then RCs degenerate into $TGDs$ and {\em Escape} degenerates into {\em Chase}.}, and instead of the game of Escape one finds the conventional notion of Chase. But, while in [G15] the counterpart 
of Lemma \ref{bridge} follows from the well-known fact that Chase is a universal structure, here we do not have such convenient tool available off-the-shelf, and we need to built our own.

\begin{lemma} \label{universal}

Suppose 
structures $\database_0$ and $\mathbb M$ over $\bar\Sigma$ are such that there exists 
a homomorphism $h_0:\database_0\rightarrow\mathbb M$. Let $T$ be a set of RCs and suppose  ${\mathbb M} \models T$.
Then from some final position $\history\in\histories(T, \database_0)$ there exists a homomorphism $h:\history\rightarrow \mathbb{M}$
\end{lemma}

\begin{proof}
First we need to prove:

\begin{lemma}
\label{lm-universal}
For structures $\database_i$, $\mathbb M$ over $\bar\Sigma$, a homomorphism $h_i :\database_i\rightarrow\mathbb{M}$ and set of RCs $T$ if ${\mathbb M} \models T$ then there exists some structure $\database_{i+1}$ such that $\step(\database_i,\database_{i+1}, T)$ and there exists homomorphism $h_{i+1}:\database_{i+1}\rightarrow\mathbb{M}$ such that $h_i \subseteq h_{i+1}$.
\end{lemma}

\begin{proof}
For $r = \pair{x,y,\rtgd{X}{Y}}$ in $R_i = \requests(T, \database_i)$ let $x' = h_i(x)$ and $y' = h_i(y)$. We know that ${\mathbb M} \models T$ so ${\mathbb M} \models Y(x',y')$ and thus for some $a_1a_2\ddd a_n \in Y$ there is path $p' = a_1(x',x'_1), \\ a_2(x'_1,x'_2)\ddd  a_n(x'_{n-1}, y')$ in ${\mathbb M}$. Let $\database_i^r$ be a structure created by adding to $\database_i$ new path $p = a_1(x,x_1),\\a_2(x_1,x_2),\ddd a_n(x_{n-1}, y)$ (with $x_i$ being new veritces). Let $h_i^r = h_i \cup \set{\pair{x_i, x_i'} | i \in [n-1]}$. Now let $\database' = \bigcup_{r \in R_i}\database_i^r$ and $h_i' = \bigcup_{r \in R_i}h_i^r$. It is easy to see that $\database_i'$ and $h_i'$ are requested $\database_{i+1}$ and $h_{i+1}$.
\end{proof}

To end the proof of Lemma~\ref{universal} notice that if $\database_0,\database_1,\ddd$ are as constructed by Lemma~\ref{lm-universal} then $\bigcup_{i=0}^{\infty} \database_i$ is equal to some final position from $\histories(T,\database_0)$ and that $\bigcup_{i=0}^{\infty} h_i$ is required homomorphism $h$.
\end{proof}

Now we will prove the (i)$\Rightarrow$(ii) part of Lemma~\ref{bridge}.

Let ${\mathbb M}$ be a counterexample from Lemma~\ref{lm-det-struct}, $a,b$ and $w\in Q_0$ such that ${\mathbb M} \models (G(w))(a,b)$ and ${\mathbb M} \not\models (R(Q_0))(a,b)$. Applying Lemma~\ref{universal} to $\database_0 = G(w[a,b])$ and 
to $\mathbb M$ we know that there exists a final position $\history$ such that there is homomorphism from $\history$ to ${\mathbb M}$. It is clear that $\history \not\models (R(Q_0))(a,b)$ as we know that ${\mathbb M} \not\models (R(Q_0))(a,b)$. This shows that $\history$ is indeed a winning final position.

This concludes the proof of the Lemma~\ref{bridge}.

\section{The Reduction}\label{redukcja}

\begin{definition}[\textbf{Our Grid Tilling Problem (OGTP)}]\label{ogtp}
Given a set of {\bf shades}  $\mathcal{S}$ (\textbf{black} $\in \mathcal{S}$) and a list $\mathcal{F} \subseteq \{V,H\} \times \mathcal{S} \times \{V,H\} \times \mathcal{S} $ of forbidden pairs $\pair{a,b}$ where $a,b \in \{V,H\} \times \mathcal{S}$ determine whether there exists a square grid $\mathbb{G}$ (a directed graph, as in Figure 1. but of any size) such that:
\begin{itemize}
\item[(a1)] each horizontal edge of $\mathbb{G}$ has a label from $\{H\} \times \mathcal{S}$
\item[(a2)] each vertical edge of $\mathbb{G}$ has a label from $\{V\} \times \mathcal{S}$
\item[(b1)] bottom-left vertical edge is colored \textbf{black}
\item[(b2)] upper-right horizontal edge is colored \textbf{black}
\item[(b3)] \textbf{G} contains no forbidden paths of length $2$ labeled by $(a,b) \in \mathcal{F}$
\end{itemize}
\end{definition}

\begin{figure}
\centering
\includegraphics[width=0.5\textwidth]{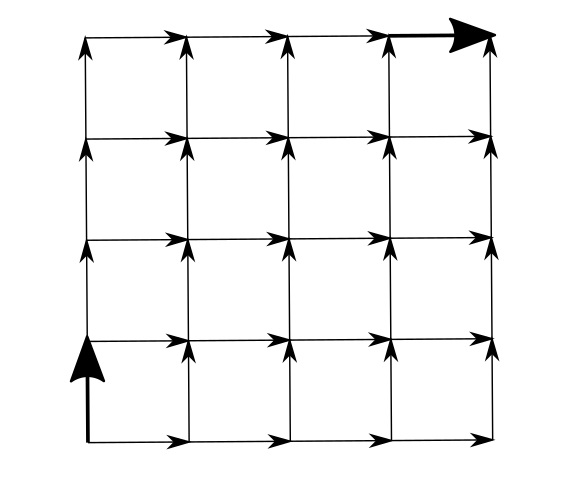}
\caption{\label{fig:fig1} Our Grid.}
\end{figure}

By standard argument one can show that:

\begin{lemma}
Our Grid Tilling Problem is undecidable.
\end{lemma}

Now we present a reduction from OGTP to the QDP-RPQ. Suppose an instance $\langle \mathcal{S},\mathcal{F}\rangle$ of OGTP is given, we will 
construct an instance $\langle \mathcal{Q}, Q_0\rangle$ of  QDP for RPQ. 

\noindent
The edge alphabet (signature) will be $\Sigma = \{\alpha, \beta, \omega\} \cup \Sigma_{0}$, 
where $\Sigma_{0} = \{A,B\} \times \{\hor, \ver\} \times \{W,C\} \times \mathcal{S}$. We think of
$H$ and $V$ as {\bf directions} -- \textit{Horizontal} and \textit{Vertical}. $W$ and $C$ stand for \textit{Warm} and \textit{Cold}.
It is worth reminding at this point that relations from $\bar\Sigma$ will -- apart from shade, direction and \textbf{temperature} -- have also {\bf color}, red or green.

\begin{notation}
\label{notation-cc}
We  use the following notation for elements of $\Sigma_{0}$:

\vspace{-4mm}
\begin{center}
$ \CC{p}{q}{r}{s}{\empty} := (\mathbf{p},q,r,s) \in \Sigma_{0} $
\end{center}
\vspace{1mm}

\noindent
Symbol $\bullet$ and empty space are to be understood as wildcards. 
This means, for example,  that $\CC{A}{\hor}{}{a}{\empty}$ denotes the set $\{\CC{A}{\hor}{W}{a}{\empty}, \CC{A}{\hor}{C}{a}{\empty} \}$ and $\CC{\bullet}{\hor}{W}{a}{\empty}$ denotes $\{\CC{A}{\hor}{W}{a}{\empty}, \CC{B}{\hor}{W}{a}{\empty} \}$. 
\end{notation}

\noindent
Now we  define $\mathcal{Q}$ and $Q_{0}$. 
Let $\mathcal{Q}_{good}$ be a set of 8 languages:
\begin{enumerate}
\item $\omega$
\item $\alpha + \beta$
\item $\CC{B}{\hor}{W}{\empty}{\empty}\CC{A}{\ver}{W}{\empty}{\empty} + \CC{B}{\ver}{C}{\empty}{\empty}\CC{A}{\hor}{C}{\empty}{\empty}$
\item $\CC{A}{\hor}{C}{\empty}{\empty}\CC{B}{\ver}{C}{\empty}{\empty} + \CC{A}{\ver}{W}{\empty}{\empty}\CC{B}{\hor}{W}{\empty}{\empty}$
\item $\CC{B}{\ver}{C}{\empty}{\empty} + \CC{B}{\ver}{W}{\empty}{\empty}$
\item $\CC{B}{\hor}{W}{\empty}{\empty} + \CC{B}{\hor}{C}{\empty}{\empty}$
\item $\CC{A}{\ver}{W}{\empty}{\empty} + \CC{A}{\ver}{C}{\empty}{\empty}$
\item $\CC{A}{\hor}{C}{\empty}{\empty} + \CC{A}{\hor}{W}{\empty}{\empty}$
\end{enumerate}


\noindent
Let $\mathcal{Q}_{bad}$ be a set of languages:

\begin{enumerate}
\item $\beta \Big( \bigoplus_{s \in \mathcal{S} \setminus \{black\}} \CC{A}{\ver}{W}{s}{\empty} \Big) \Sigma_{0}^{\star} \omega$
\item $\beta \Sigma_{0}^{\star} \Big( \bigoplus_{s \in \mathcal{S} \setminus \{black\}} \CC{B}{\hor}{W}{s}{\empty} \Big) \omega$
\item $\beta \Sigma_{0}^{\star} \CC{\bullet}{d}{W}{a}{\empty} \CC{\bullet}{d'}{W}{b}{\empty} \Sigma_{0}^{\star} \omega$ for each forbidden  $\langle(d,a), (d',b)\rangle \in \mathcal{F}$.
\end{enumerate}

\noindent
Finally, let $\mathcal{Q}_{ugly}$ be a set of languages:

\begin{enumerate}
\item $\alpha \Sigma_{0}^{\star} \CC{\bullet}{\empty}{W}{\empty}{\empty} \Sigma_{0}^{\star} \omega$
\item $\beta \Sigma_{0}^{\star} \CC{\bullet}{\empty}{C}{\empty}{\empty} \Sigma_{0}^{\star} \omega$
\end{enumerate}

\noindent

We write ${Q}_{good}^{i}, {Q}_{bad}^{i}, {Q}_{ugly}^{i}$ to denote the i-th language of the corresponding group. Now we can define 
$$\mathcal{Q} := \mathcal{Q}_{good} \cup \mathcal{Q}_{bad} \cup \mathcal{Q}_{ugly}$$

The sense of the construction will (hopefully) become clear later. But already at this point the reader can 
notice that there is a fundamental difference between languages from  ${\mathcal Q}_{good}$  and  languages from ${\mathcal Q}_{bad} \cup{\mathcal Q}_{ugly}$. Languages from 
${\mathcal Q}_{good}$ are all finite. The regular constraints 
$(Q^3_{good})^\leftrightarrow$ and $(Q^4_{good})^\leftrightarrow$ are of the form  ``for vertices $x,y,z$  and edges $e_1(x,y)$ and $e_2(y,z)$ of some color in the current structure, create a new $y'$ and add edges  $e'_1(x,y')$ and $e'_2(y',z)$ of the opposite color''
where the pair $\pair{e_1,e_2}$ comes from some small finite set of possible choices. Satisfying requests generated by the remaining languages in ${\mathcal Q}_{good}$  do not even allow/require 
adding a new vertex $y'$ -- just one new edge is added. 

On the other hand, each language in   ${\mathcal Q}_{bad} \cup{\mathcal Q}_{ugly}$ contains infinitely many words -- all 
words with some bad or ugly pattern. For $L\in {\mathcal Q}_{bad} \cup{\mathcal Q}_{ugly}$ requests generated by $L$ are of the form ``if you have any path in the current structure, green or red, between some verticies $x$ and $y$, containing such pattern, then add any new path from $x$ to $y$, of the opposite color, also containing the same pattern''. 

A small difference between languages in ${\mathcal Q}_{bad}$ and in ${\mathcal Q}_{ugly}$ is that 
languages in  ${\mathcal Q}_{ugly}$ do not depend on the constraints from the instance of Our Grid Tiling Problem while ones in ${\mathcal Q}_{bad}$ encode this instance. One important difference between languages in  ${\mathcal Q}_{good}\cup {\mathcal Q}_{ugly}$ and  ${\mathcal Q}_{bad}$ is that only  the last do mention shades.

Finally, define  $Q_{start} := \alpha [ \CC{A}{\hor}{C}{\empty}{\empty}\CC{B}{\ver}{C}{\empty}{\empty} ]^{+} \omega$, and let: 
$$Q_{0} := Q_{start} + \bigoplus_{L \in \mathcal{Q}_{ugly}} {L} + \bigoplus_{L \in \mathcal{Q}_{bad}} {L}$$

\section{The structure  of the proof of correctness}\label{przewodnik}
To end the proof of Theorem \ref{i1} we need to prove:

\begin{lemma}\label{main-lemma}\label{poprawnoscredukcji}
The following two conditions are equivalent:
\begin{enumerate}[(i)]
\item \label{grid-cond} An instance $\langle {\mathcal{S}},{\mathcal{F}} \rangle$ of OGTP has no solution.
\item \label{det-cond} $\mathcal{Q}$ determine $Q_{0}$.
\end{enumerate}
\end{lemma}

For the (\ref{grid-cond})~$\Rightarrow$~(\ref{det-cond}) implication we will employ Lemma~\ref{bridge}, showing that if the instance $\langle {\mathcal{S}},{\mathcal{F}} \rangle$ has no solution then
{\em Fugitive} does not have a winning strategy in the Escape($\mathcal{Q}$, $Q_{0}$). 
As we remember from Section \ref{game}, in such a game {\em Fugitive} will first choose, as the initial position of the game,  a structure $w[a,b]$ 
for some $w\in G(Q_{0})$. Then, in each step, he will identify all the requests present in the current structure and satisfy them. He will win if he will be able to play forever without satisfying the query    $(R(Q_{0}))(a,b)$. 

While analyzing the strategy of {\em Fugitive} we will use the words ``must not'' and ``must'' as shorthands for
``or otherwise he will quickly lose the game''. 

Now our plan is first to notice that in his strategy {\em Fugitive} must obey the following principles:\smallskip\\
(I)~ The structure resulting from his initial move must be $(G(w))[a,b]$ for some $w\in Q_{start}$.\\
(II)~He must never allow any request generated by $\mathcal{Q}_{bad} \cup \mathcal{Q}_{ugly}$ to form in the current structure.
Notice that if no such words ever occur in the structure then all the requests 
are generated by  languages from $\mathcal{Q}_{good}$.

Then we will assume that {\em Fugitive's} play indeed follows the two principles and we will imagine us watching him playing, but watching in special glasses that make us insensitive to the shades from $\mathcal{S}$.
Notice that, since the only requests {\em Fugitive} will satisfy, are from $\mathcal{Q}_{good}$, we will not miss anything -- as the definitions of languages in $\mathcal{Q}_{good}$ are themselves shade-insensitive. 
In Section~\ref{straszna} we will prove that {\em Fugitive} must construct some particular structure, defined earlier in Section~\ref{mgrid} and called 
${\mathbb G}_m$, for some $m\in\mathbb N$. Then, in a short Section \ref{ulga} we will take off our glasses and recall that  the edges of ${\mathbb G}_m$ actually have shades.
Assuming that the original instance of OGTP has no solution, we will get that $R(\mathcal{Q}_{bad})(a,b)$  holds in the constructed structure.  This will end the proof of the (i)$\Rightarrow$(ii) direction. For the implication ($\neg$i)$\Rightarrow$($\neg$ii) we will notice, again in Section \ref{ulga} that if $\langle {\mathcal{S}},{\mathcal{F}} \rangle$ has a solution, then one of the structures ${\mathbb G}_m$, with shades duly assigned to edges, forms a counterexample $\mathbb M$ as required by Lemma \ref{lm-det-struct}. Since this  $\mathbb M$ will be finite, we will show that if  the instance 
$\langle {\mathcal{S}},{\mathcal{F}} \rangle$
of OGTP has a  solution, then ${\mathcal Q}$ does not
finitely determine $Q_0$ (which is a stronger statement than just saying that ${\mathcal Q}$ does not determine $Q_0$).


\section{Principle I : $\database_{0}$} 

The rules of the game of Escape are such that {\em Fugitive} loses when
he builds a path (from $a$ to $b$) labeled with $w\in R(Q_0)$. So -- when
trying to encode something --  one can think of words in $Q_0$ as of some sort of
forbidden patterns. And thus one can think of  $Q_0$ as of
a tool detecting that the player is cheating and not really building a
valid computation of the computing device we encode. Having this in
mind the Reader can imagine why the
words from languages from the groups
 ${\mathcal Q}_{bad}$  and  ${\mathcal Q}_{ugly}$,  which clearly are all about suspiciously looking patterns,
 are all in $Q_0$

But  another  rule of the game is that at the beginning  {\em Fugitive}
picks his initial position ${\mathbb D}_0$ as a path  (from $a$ to $b$)
labeled with some $w\in G(Q_0)$, so it would be nice to think of $Q_0$
as of initial configurations of this computing device. The fact that the same object 
is playing the set of forbidden patterns and, at the same time, the set of initial configurations
is a problem.  But this problem is solvable, as we are going to show in this Section. And having the languages 
  ${\mathcal Q}_{bad} \cup{\mathcal Q}_{ugly}$  also in $Q_0$
 is part of the solution. 

Assume that $\history$ is a final position of a play of the \textit{Escape} game that started with  $\database_0 = G(w)[a,b]$ for some $w \in Q_0$. This means, by Lemma \ref{rcsatisfied}, that $\history \models \mathcal{Q}^\leftrightarrow$. Recall that $\history$ is a structure over $\bar\Sigma$, which means that each edge of $\history$ is either red or green.

%
%
%
%
%

\begin{observation}\label{nogreenQ}
For all $x,y \in \history$ if $\history \models G(L)(x,y)$ for some $L \in \mathcal{Q}_{ugly} \cup \mathcal{Q}_{bad}$ then $\history \models R(Q_{0})(x,y)$.
\end{observation}

\begin{proof}
Notice that $\rtgd{G(L)}{R(L)} \in \mathcal{Q}^\rightarrow$ so $\history \models \Path{R(L)}{x,y}$ and as $L \subseteq Q_{0}$ it follows that $\history \models \Path{R(Q_{0})}{x,y}$.
\end{proof}

\begin{lemma}[\textbf{Principle I}]
Fugitive must choose to start the \textit{Escape} game from $\database_{0} = G(q)[a,b]$ for $q \in Q_{start}$.
\end{lemma}

\begin{proof}
If $q \in Q_{0} \setminus Q_{start}$ then $\database_{0} \models \Path{G(L)}{a,b}$ for some $L \in \mathcal{Q}_{ugly} \cup \mathcal{Q}_{bad}$ and it follows from Observation \ref{nogreenQ}. that \textit{Fugitive} loses.
\end{proof}

\begin{figure*}
\centering
\includegraphics[width=1.0\textwidth]{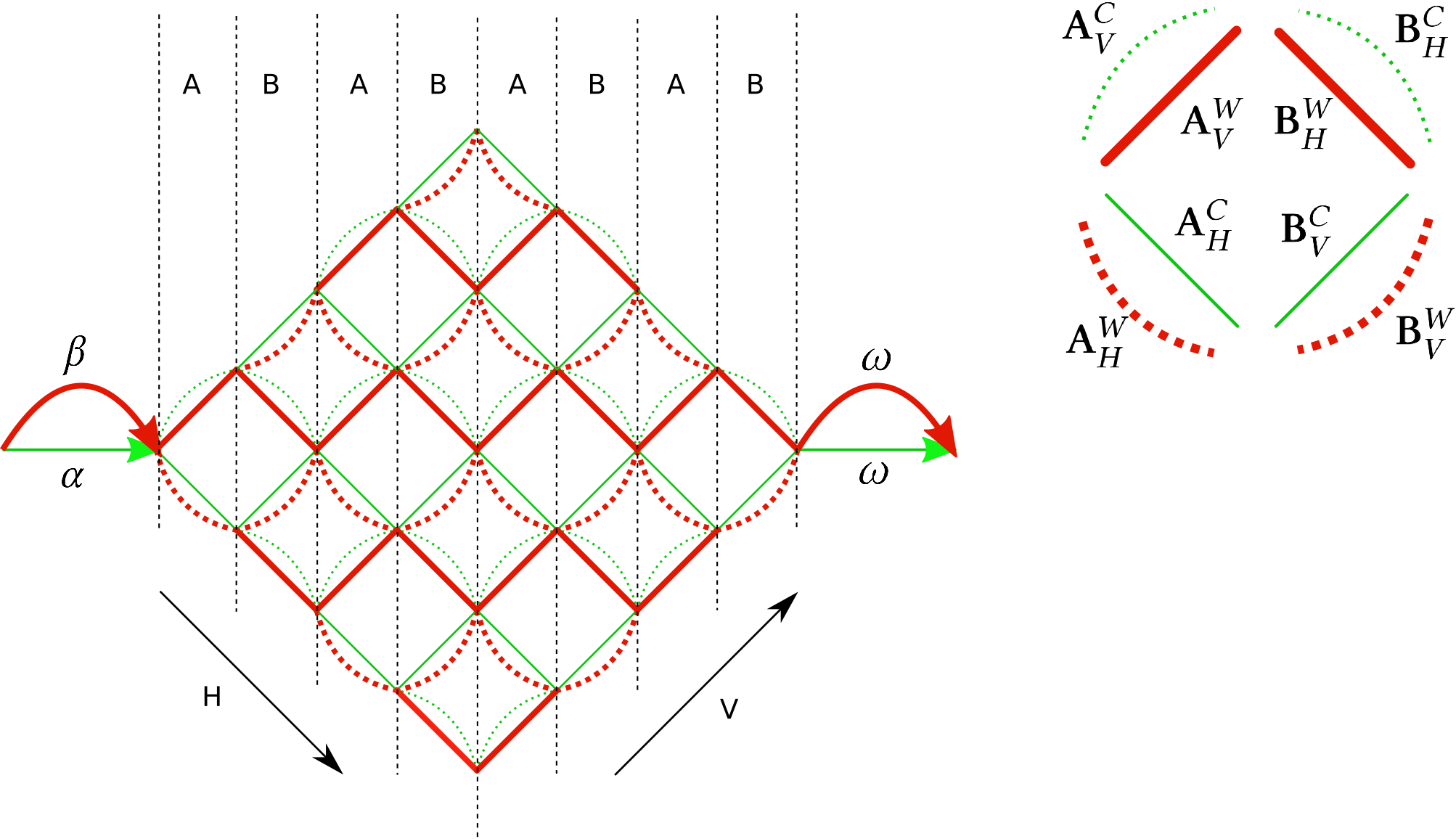}
\caption{\label{fig:fig1} ${\mathbb G}_m$ with $m=4$ (left). Smaller picture in the top-right corner explains how the different line styles on the main picture map to $\Sigma_{0}$.\protect\footnotemark}
\end{figure*}


\section{The grid ${\mathbb G}_m$}\label{mgrid}

\begin{definition}
${\mathbb G}_m$, for $m \in \mathbb{N}$, is (see Fig. 2)  a directed graph $ (V,E)$  where

 $V = \{a,b\} \cup \{v_{i,j} : i,j \in \interval{0}{m} \}$ and where the 
 edges from $E$ are labeled with symbols  $\alpha$ or $\beta$ or $\omega$ or one of the symbols of the form $\CC{p}{q}{r}{\empty}{\empty}$, where -- like before -- $p\in \{A,B\}$, $q\in \{\hor, \ver\}$  and $r\in \{W,C\}$. Each label has to also be either red or green (this gives us $(3+2^3)2$ possible labels, but only 12 of them will be used).
Notice that there is no $s\in \mathcal S$ here: the labels we now use are sets of symbols from $\bar\Sigma$ like in Notation~\ref{notation-cc}. One should imagine that we watch {\em Fugitive's} play in 
shade filtering glasses.

The edges of ${\mathbb G}_m$ are as follows:
\begin{itemize}
\item Vertex $v_{0,0}$ is a successor of $a$. Vertex $b$ is a successor of $v_{m,m}$. The successors of $v_{i,j}$ are  $v_{i+1,j}$ and $v_{i,j+1}$ (if they exist). Each node is connected to each of its successors with two edges, one green and one red.

 \item Each ``Cold'' edge, labeled with a symbol in $\CC{\bullet}{\empty}{C}{\empty}{\empty}$, is green.
 
  \item Each ``Warm'' edge, labeled with a symbol in   $ \CC{\bullet}{\empty}{W}{\empty}{\empty}$, is red.
  
  \item Each edge  $\langle v_{i,j}, v_{i+1,j}\rangle$ is horizontal -- its label is from $ \CC{\bullet}{\hor}{\empty}{\empty}{\empty}$.
  
 \item Each edge $\langle v_{i,j},v_{i,j+1}\rangle$ is vertical-- its label is from $\CC{\bullet}{\ver}{\empty}{\empty}{\empty}$.
 
 \item The label of each edge leaving $v_{i,j}\neq v_{m,m}$, with $i+j$ even, is from $ \CC{A}{\empty}{\empty}{\empty}{\empty}$, the label of each edge leaving $v_{i,j}\neq v_{m,m}$, with $i+j$ odd, is from $ \CC{B}{\empty}{\empty}{\empty}{\empty}$.
 
  \item Edges  $(a,v_{0,0},G(\alpha))$ and  $(a,v_{0,0},R(\beta))$ are in $ E$.
  \item Edges  $(v_{m,m},b,G(\omega))$ and $ (v_{m,m},b,R(\omega))$ are in $ E$.
 
\end{itemize}


\end{definition}


\footnotetext{Please use a color printer if you can.}

\section{Principle II} 

In this section we assume that the \textit{Fugitive} obeys Principle I and he selects the initial structure $\database_{0} =  G(\alpha [ \CC{A}{\hor}{C}{\empty}{\empty}\CC{B}{\ver}{C}{\empty}{\empty} ]^{m} \omega)[a,b]$ for some $m$.

\begin{lemma}\label{alphabeta}
Suppose $\history$ is the final position of a play of the Escape game which started from $\database_{0}$.
\begin{enumerate}
\item Every edge $e \in \history$ labeled with $G(\alpha), R(\alpha), G(\beta)$ or $R(\beta)$ begins in $a$.
\item Every edge $e \in \history$ labeled with $G(\omega)$ or $R(\omega)$ ends in $b$.
\end{enumerate}
\end{lemma}

\begin{proof}
(1) By induction we show that the claim is true in every $\database_i$. It is clearly true in $\database_{0}$. 
For the induction step use the fact  that for every language $L \in \mathcal{Q}$ and for each word $w \in L$  if $w$ contains $\alpha$ or $\beta$ then:\\
-- this  $\alpha$ or $\beta$ is the first letter of $w$ and \\
-- all words in $L$ begin from  $\alpha$ or $\beta$.\\
(2) Analogous.
\end{proof}

\begin{lemma}[\textbf{Principle II}]
Fugitive must never allow any request generated by $\mathcal{Q}_{bad}$ and $\mathcal{Q}_{ugly}$ to form in the current structure.
\end{lemma}

\begin{proof}
Let $\database$ be the current structure and $L \in \mathcal{Q}_{bad} \cup \mathcal{Q}_{ugly}$. 

First assume that $\database \models R(L)(x,y)$ for some $x,y$. Notice that from Lemma~\ref{alphabeta} $x = a$ and $y = b$. Because of that $\database \models R(L)(a,b)$ which means that $\database \models R(Q_{0})(a,b)$ and \textit{Fugitive} loses.

Now assume that $\database \models G(L)(x,y)$ for some $x,y$. Similarly, from Lemma~\ref{alphabeta}, $x = a$ and $y = b$. We have that $\pair{a,b,L^{\rightarrow}} \in rq(\mathcal{Q}^{\leftrightarrow},\database)$ so \textit{Fugitive} must satisfy this request with $R(w)[a,b]$ for some $w \in L$ and he loses, since $L \subseteq Q_{0}$.
\end{proof}


\begin{figure*}

\includegraphics[width=1.0\textwidth]{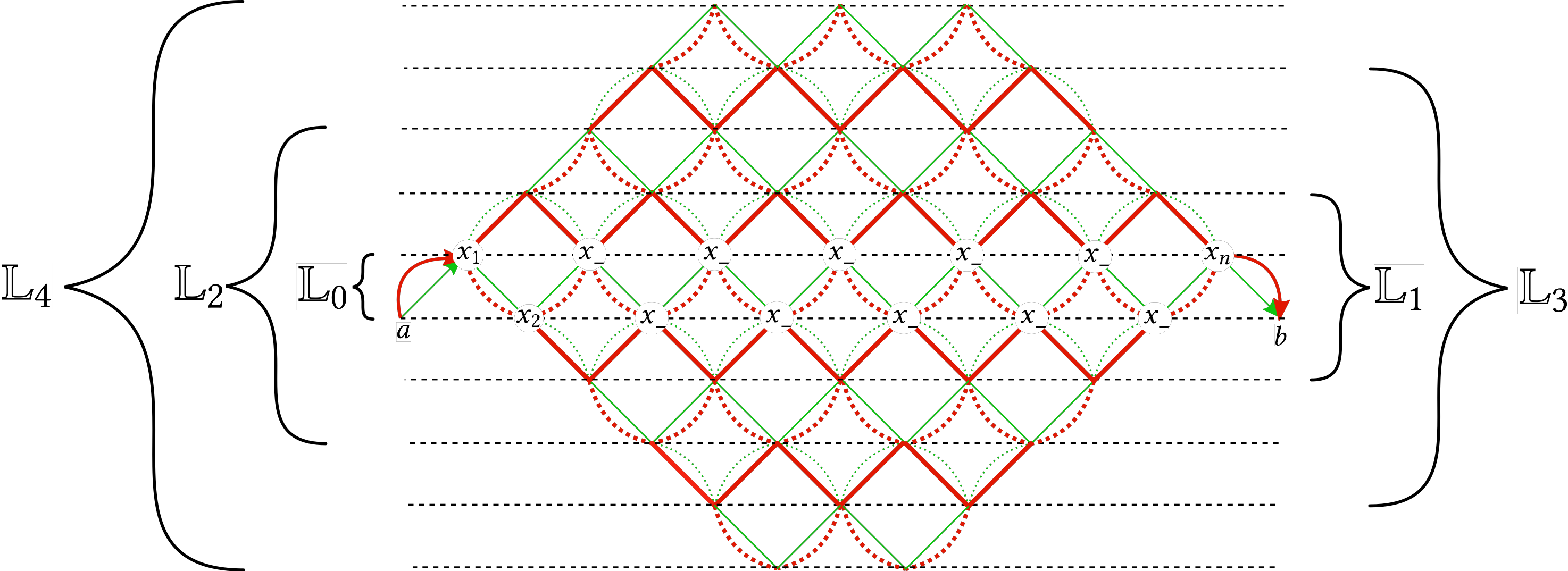}
\caption{\label{fig:fig2} Five first Layers of ${\mathbb G}_m$ with $m=6$.}
\end{figure*}

\section{Now we do not see the shades}\label{straszna}

As we already said, now we are going to watch, and analyze, {\em Fugitive's} play in shade filtering glasses. We assume he obeys Principle I, otherwise he would lose. We also assume he obeys Principle II, but wearing our glasses we are not able to tell 
whether any word from $G(\mathcal{Q}_{bad}) \cup  R(\mathcal{Q}_{bad})$ occurs in the current structure. For this reason we cannot use, 
 in our analysis, arguments referring to languages in ${\mathcal Q}_{bad}$. We are however free to use arguments from Principle II, referring to languages in ${\mathcal Q}_{ugly}$.

\begin{lemma}\label{musi-G}
Suppose in his initial move {\em Fugitive} selects  $\database_{0} = G(\alpha [ \CC{A}{\hor}{C}{\empty}{\empty}\CC{B}{\ver}{C}{\empty}{\empty} ]^{m} \omega)[a,b]$ . Then the final position $\mathbb H$  must be equal (from the point of view of a shades-insensitive spectator) to $\mathbb{G}_m$
\end{lemma}

To prove Lemma~\ref{musi-G} it is enough to prove the following Lemma:

\begin{lemma}\label{musi-warstwy}
Let $\layer_i$ be like on Figure~\ref{fig:fig2}  and  $\layer^G_i$ and $\layer^R_i$ be parts of $\layer_i$ consisting of (resp.) green and red edges. Then: 

\begin{enumerate}[(i)]
\item $\database_0 = \layer^G_{0}$,
\item $\database_{2i} = \layer^G_{2i} \cup \layer_{2i - 1}$,
\item $\database_{2i + 1} = \layer^R_{2i + 1} \cup \layer_{2i}$.
\end{enumerate}
\end{lemma}

\noindent Lemma \ref{musi-warstwy} (i) is Principle I restated. 
Next subsections of this Section are devoted to the proof of Lemma \ref{musi-warstwy} (ii) and (iii).

\subsection{General rules for the Fugitive}

Now assume $\database_{0}$ as demanded by Lemma \ref{musi-G} was really selected and denote vertices of this $\database_{0}$ by $a,x_{1}, \dots, x_{n},b$, with $n = 2m+1$ (see Figure 3).

\begin{lemma}
For every final position $\history$ that was built obeying Principles I and II:
\begin{enumerate}
\item Every edge $e \in \history$ labeled with $G(\alpha), R(\alpha), G(\beta)$ or $R(\beta)$ connects $a$ and $x_{1}$.
\item Every edge $e \in \history$ labeled with $G(\omega)$ or $R(\omega)$ connects $x_{n}$ and $b$.
\end{enumerate}
\end{lemma}

\begin{proof}
Notice that by Principle II there were no requests formed by either $\mathcal{Q}_{bad}$ or $\mathcal{Q}_{ugly}$ during the  game that led to $\history$. It means that all requests were generated by $\mathcal{Q}_{good}$. But  for every language $L \in \mathcal{Q}_{good}$ for each $w \in L$  if $w$ contains $\alpha, \beta$ or $\omega$ then $w$ is a one letter word, and also all other words of this language contain 
one letter. So satisfying a request involving $\alpha$, $\beta$ or $\omega$ never requires creating new vertices.
\end{proof}

\begin{lemma}\label{allpaths}
For each $y \in \history, y \neq a$ there exist, in $\history$:

\begin{itemize}
\item a red path from $x_{1}$ to $y$,
\item a green path from $x_{1}$ to $y$,
\end{itemize}

For each $y \in \history, y \neq b$ there exist, in $\history$:
\begin{itemize}
\item a red path from $y$ to $x_{n}$,
\item a green path from $y$ to $x_{n}$.
\end{itemize}
\end{lemma}

\begin{proof}
Notice that for each $c \in \Sigma_{0}$ there exists a language $L \in \mathcal{Q}_{good}$ such that $c \in L$. This means that for all $u,w \in \history$ such that these vertices are endpoints of a green edge $e = (u,w,G(c)), c \in \Sigma_{0}$ there is also a red path connecting $u$ and $w$ $\in \history$ 
(this is since $\history\models {\mathcal Q}_{good}^{\ \leftrightarrow}$ )

Reasoning for red edges is analogous.
\end{proof}

In his first move  \textit{Fugitive} must 
satisfy all the requests in $S_0=rq({\mathcal Q}^\leftrightarrow, \database_{0})$. Notice that (since all the edges of $\database_{0}$ are green and there are no bad or ugly patterns in $\database_{0}$)
all requests in $S_0$ 
are actually generated by RCs in ${\mathcal Q}_{good}^{\ \rightarrow}$.
And one of them is generated by $(Q_{good}^2)^{\ \rightarrow}$. Next lemma does not look 
spectacular, but this is how we get our foot in the door:

\begin{lemma}
Request $req = \pair{a,x_{1},(\alpha + \beta)^{\rightarrow}}$ in $S_{0}$ must be satisfied with $R(\beta)[a,x_{1}]$.
\end{lemma}

\begin{proof}
First notice that there are numerous requests in $S_0$ generated by $Q_{good}^{4}$, all of them of the form $\pair{x_{i},x_{i+2},Q_{good}^{4 \ \rightarrow}}$. Each of them can potentially be satisfied in one of two  ways: either by adding a new path labeled  with a word $R(\CC{A}{\ver}{W}{\empty}{\empty}\CC{B}{\hor}{W}{\empty}{\empty})$ from $x_{i},x_{i+2}$ or by adding a new path labeled with $R(\CC{A}{\hor}{C}{\empty}{\empty}\CC{B}{\ver}{C}{\empty}{\empty})$.

Consider what would happen if {\em Fugitive} tried to satisfy $req$ with $R(\alpha)$ instead of $R(\beta)$. First assume that there exists $req \in S_0$ generated by $Q_{good}^{4}$ that is satisfied with $R(\CC{A}{\ver}{W}{\empty}{\empty}\CC{B}{\hor}{W}{\empty}{\empty})$. Then $\database_{1} \models R(Q_{ugly}^{1})(a,b)$ and this is forbidden by Principle II. So all requests in $S_0$ generated by $Q_{good}^{4}$ must be satisfied with $R(\CC{A}{\hor}{C}{\empty}{\empty}\CC{B}{\ver}{C}{\empty}{\empty})$. But then $\database_{1} \models R(Q_{start})(a,b)$ and  \textit{Fugitive} loses.
\end{proof}

Now we know that, alongside the green $\alpha$, there must exist the red $\beta$ leading to $x_1$
(see Figure 2). From this we get that:

\begin{lemma}\label{redWgreenC}
If \ $\history$ is a final position  that was built obeying Principles I and II (which started with $\database_{0}$) then: for each edge $e \in \history$,
\begin{enumerate}
\item $e$ is labeled with $c \in R(\Sigma_{0}) \Leftrightarrow c \in R\CC{\bullet}{\empty}{W}{\empty}{\empty}$
\item $e$ is labeled with $c \in G(\Sigma_{0}) \Leftrightarrow c \in G\CC{\bullet}{\empty}{C}{\empty}{\empty}$
\end{enumerate}
\end{lemma}

\begin{proof}
(1) Assume by contradiction that there exists a red edge $e\in \history$, from some $x$ to some $x'$, labeled with $c \in R\CC{\bullet}{\empty}{C}{\empty}{\empty}$. By Lemma~\ref{allpaths} there is a path, consisting of edges from 
$R(\Sigma_0)$, from $x_1$ to $x$ and another such path from $x'$ to $x_n$. 
This implies that $\history \models Q_{ugly}^{2}(a,b)$ which is forbidden by Principle II.
(2) Like (1) but then $\history \models Q_{ugly}^{1}(a,b)$.
\end{proof}

Notice that each $Q_{good}^{i}$ for $i = 3 \dots 8$ consists of two words (from the point of view of a shades-insensitive spectator). This sounds like good news for {\em Fugitive}: when satisfying requests 
generated by these languages he has some choice. But actually he does not, as the next lemma tells us:

\begin{lemma}\label{changeword}
Let $i\in \{3 \dots 8\}$ and let $Q_{good}^{i} = \{w_{i},w_{i}'\}$.
\begin{enumerate}
\item  If  $\database_{j} \models G(w_{i})(x,y)$, for some $j$, and $\database_{j} \not\models R(Q_{good}^{i})(x,y)$ then $\pair{x,y,Q_{good}^{i \ \rightarrow}} \in rq(Q_{good}^{i \ \rightarrow},\database_{j})$ and the Fugitive must satisfy this request with $R(w_{i}')[x,y]$.
\item  If  $\database_{j} \models R(w_{i})(x,y)$, for some $j$, and $\database_{j} \not\models G(Q_{good}^{i})(x,y)$ then $\pair{x,y,Q_{good}^{i \ \leftarrow}} \in rq(Q_{good}^{i \ \leftarrow},\database_{j})$ and the Fugitive must satisfy this request with $G(w_{i}')[x,y]$.
\end{enumerate}
\end{lemma}

\begin{proof}
(1) Let $i \in \{3,\dots,8\}$ and let $j$ be such that $\database_{j} \models G(w_{i})(x,y)$ and $\database_{j} \not\models R(Q_{good}^{i})(x,y)$. Assume by contradiction that \textit{Fugitive} satisfies $\pair{x,y,Q_{good}^{i \ \rightarrow}}$ with $R(w_{i})[x,y]$. Then $\database_{j+1} \models G(w_{i})(x,y)$ and $\database_{j+1} \models R(w_{i})(x,y)$. Let $c$ be any letter of $w_{i}$ (notice that $c \in \Sigma_{0}$). We have that there exist vertices $u,w,p,q \in \database_{j+1}$ such that $\database_{j+1} \models G(c)(u,w)$ and $\database_{j+1} \models R(c)(p,q)$ and this  contradicts Lemma \ref{redWgreenC}.
(2) Analogous to the proof of (1).
\end{proof}

Now, in Section \ref{nieparzysty} we
 assume that $\database_{2i} = \layer^G_{2i} \cup \layer_{2i - 1}$ and show that $\database_{2i + 1}$ is as claimed in Lemma~\ref{musi-warstwy} (ii) and  in Section \ref{parzysty} we assume that $\database_{2i + 1} = \layer^R_{2i + 1} \cup \layer_{2i}$ and show that $\database_{2i + 2}$ is as claimed in Lemma~\ref{musi-warstwy} (iii).
\subsection{Fugitive's move 2$i$: from $\database_{2i}$ to  $\database_{2i+1}$}\label{nieparzysty}

\begin{observation}\label{odod}
For $\database_{2i}$ it is true that:
\begin{enumerate}[(1)]
	\item All requests in $\database_{2i}$ generated by $Q_{good}^{4}$ must be satisfied with $R(\CC{A}{\ver}{W}{\empty}{\empty}\CC{B}{\hor}{W}{\empty}{\empty})$.
	\item All request in $\database_{2i}$ generated by $Q_{good}^{3}$ must be satisfied with $R(\CC{B}{\hor}{W}{\empty}{\empty}\CC{A}{\ver}{W}{\empty}{\empty})$
	\item All requests in $\database_{2i}$ generated by $Q_{good}^{5}$ must be satisfied with $R\CC{B}{\ver}{W}{\empty}{\empty}$.
	\item All requests in $\database_{2i}$ generated by $Q_{good}^{8}$ must be satisfied with $R\CC{A}{\hor}{W}{\empty}{\empty}$.
\end{enumerate}
\end{observation}

\begin{proof}
For (1). By hypothesis all requests that are generated by $Q_{good}^{4}$ in $\database_{2i}$ are of the form $\pair{x,y,G(\CC{A}{\hor}{C}{\empty}{\empty}\CC{B}{\ver}{C}{\empty}{\empty}) \rightarrow R(Q_{good}^{4})}$ (Note that $\CC{A}{\hor}{C}{\empty}{\empty}\CC{B}{\ver}{C}{\empty}{\empty} \in Q_{good}^4$). By Lemma \ref{changeword} \textit{Fugitive} must satisfy all such requests with $R(\CC{A}{\ver}{W}{\empty}{\empty}\CC{B}{\hor}{W}{\empty}{\empty})$. Rest of the proofs for (2)-(4) are analogous.
\end{proof}

\subsection{Fugitive's  move $2i+1$: from $\database_{2i+1}$ to  $\database_{2i+2}$}\label{parzysty}
%
%
%

Proof of the following Observation is analogous to the one of Observation~\ref{odod}.

\begin{observation}
For $\database_{2i+1}$ it is true that:
\begin{enumerate}
\item All requests in $\database_{2i+1}$ generated by $Q_{good}^{4}$ must be satisfied with $G(\CC{A}{\hor}{C}{\empty}{\empty}\CC{B}{\ver}{C}{\empty}{\empty})$.
\item All request in $\database_{2i+1}$ generated by $Q_{good}^{3}$ must be satisfied with $G(\CC{B}{\ver}{C}{\empty}{\empty}\CC{A}{\hor}{C}{\empty}{\empty})$
\item All requests in $\database_{2i+1}$ generated by $Q_{good}^{7}$ must be satisfied with $G\CC{A}{\ver}{C}{\empty}{\empty}$.
\item All requests in $\database_{2i+1}$ generated by $Q_{good}^{6}$ must be satisfied with $G\CC{B}{\hor}{C}{\empty}{\empty}$.
\end{enumerate}
\end{observation}

\subsection{The end. No more requests!}

Now it is straightforward to verify that:

\begin{observation}\label{ostatnia}
All requests generated by $\mathcal{Q}_{good}$ are already satisfied in $\database_{m+1} = \mathbb{G}_{m}$.
\end{observation}


\section{And now we see the shades again}\label{ulga}

Now we are ready to finish the proof of Lemma \ref{poprawnoscredukcji}.

Suppose the {Fugitive's} play ended, in some final position ${\mathbb H}={\mathbb G}_m$. We take off our glasses, and not only we still see this ${\mathbb H}$, but now we see it in full colors, with each edge (apart from edges labeled with $\alpha$, $\beta$ and $\omega$) having one of the shades from $\mathcal S$. Assume that 
the original instance $ \mathcal S, F$ of Our Grid Tiling Problem has no solution, and concentrate on the red edges of  ${\mathbb H}$. They form a square grid, with each vertical edge labeled with $V$, each horizontal edge labeled with $H$, and with each edge labeled with a shade  from $\mathcal S$. So clearly, one of the conditions (b1)-(b3) of Definition \ref{ogtp} is unsatisfied.  
But this implies that a path labeled with a word from one of the languages $Q_{bad}^1$-- $Q_{bad}^3$ occurs in $\mathbb H$, which is in breach of Principle II.  This ends the proof of Lemma \ref{poprawnoscredukcji} (i)$\rightarrow$ (ii).

For the proof
Lemma \ref{poprawnoscredukcji} ($\neg$i)$\rightarrow$ ($\neg$ii) assume the original instance $\langle {\mathcal S}, {\mathcal F}\rangle$ of Our Grid Tiling Problem has a solution -- a labeled grid $m\times m$ for some $m$. Call this grid ${\mathbb G}$.

Recall that 
${\mathbb G}_m$ satisfies all 
regular constraints from ${\mathcal Q}_{good}^\leftrightarrow$ (Observation~\ref{ostatnia}) and  from ${\mathcal Q}_{ugly}^\leftrightarrow $ (for trivial reasons, as no paths from any $G(L)\cup R(L)$ with $L\in {\mathcal Q}_{ugly}$ occur in ${\mathbb G}_m$). 
Now copy the shades of the edges of ${\mathbb G}$ to the respective  edges of ${\mathbb G}_m$. Call this new structure (${\mathbb G}_m$ with shades added) $\mathbb M$. It is easy to see that 
$\mathbb M$ constitutes a counterexample, as in Lemma \ref{lm-det-struct}.\\ 



~

\noindent
{\sc References}\smallskip

\noindent
[AV97] S. Abiteboul and V. Vianu, {\em Regular path queries with constraints;} Proc. of the 16th
PODS, pp. 122--133, 1997;

\noindent
[A11]
 F. N. Afrati, {\em Determinacy and query rewriting for conjunctive
queries and views}; Th.Comp.Sci. 412(11):1005--1021, March
2011;

\noindent
[AG08] R. Angles, C. Gutierrez, {\em Survey of Graph Database Models};
 ACM Comp. Surveys 
Vol. 40 Issue 1, February 2008;

\noindent
[B13]\hspace{0.5mm}P.Barceló, \hspace{-1.5mm}  {\em Querying graph databases. \hspace{-0.7mm} Simple Paths Semantics vs. Arbitrary Path Semantics}; \hspace{-0.7mm}  PODS 2013, pp. \hspace{-1mm}  175-188;

\noindent
[CMW87] I. F. Cruz, A. O. Mendelzon, and P. T. Wood, {\em A graphical query language supporting
recursion};  Proc. of ACM SIGMOD Conf. on Management of Data, 1987;

\noindent
[CGL98] D. Calvanese, G. De Giacomo, and M. Lenzerini, {\em On the decidability of query containment
under constraints}; in Proc. of the 17th  PODS,'' pp. 149--158, 1998;

\noindent
[CGLV00] D. Calvanese, G. De Giacomo, M. Lenzerini,  M.Y. Vardi.
{\em Answering regular path queries using views};  Proc..
16th Int. Conf. on Data Engineering, pages 389--398, IEEE, 2000;

\noindent
[CGLV00a] D. Calvanese, G. De Giacomo, M. Lenzerini,  M. Y. Vardi.
{\em View-based query processing and constraint satisfaction}; Proc. of 15th  IEEE LICS, 2000;

\noindent
[CGLV02] D. Calvanese, G. De Giacomo, M. Lenzerini, M.Y. Vardi.
{\em Lossless regular views};  Proc. of the  21st PODS,  pages 247--258, 2002;

\noindent
[CGLV02a] D. Calvanese, G. De Giacomo, M. Lenzerini, and M.Y. Vardi.
{\em Rewriting of regular expressions and regular path queries}; Journal of Comp. and
System Sc., 64:443--465, 2002;

\noindent
[DPT99]
 A. Deutsch, L. Popa, and Val Tannen, {\em Physical data
independence, constraints, and optimization with universal plans};
Proc. of 25th  VLDB, pages 459--
470, 1999;

\noindent
 [F15] Nadime Francis, PhD thesis, ENS de Cachan, 2015;

\noindent
 [F17]\hspace{0.5mm}N.Francis; {\em 
Asymptotic Determinacy of Path Queries Using Union-of-Paths Views}; Th.Comp.Syst. 61(1):156-190 (2017);

\noindent
[FG12]
 E. Franconi and P. Guagliardo {\em The view update problem
revisited}  CoRR, abs/1211.3016, 2012;

\noindent
[FGZ12]
 Wenfei Fan, F. Geerts, and Lixiao Zheng, {\em View determinacy for
preserving selected information in data transformations}; Inf. Syst.,
37(1):1--12, March 2012;

\noindent
[FLS98] D. Florescu, A. Levy, and D. Suciu, {\em Query containment for conjunctive queries with
regular expressions};  Proc. of the 17th  PODS,'' pp. 139--148, 1998;

\noindent
[FV98] T. Feder and M. Y. Vardi, {\em The computational structure of monotone monadic
SNP and constraint satisfaction: A study through datalog and group theory}; SIAM
Journal on Computing, 28(1):57--104, 1998;

\noindent
[FSS14] N. Francis, L. Segoufin, C. Sirangelo {\em Datalog rewritings of regular
path queries using views};  Proc. of ICDT,  pp 107--118, 2014;

\noindent
 [GB14] M. Guarnieri,
D. Basin, {\em Optimal Security-Aware Query Processing};
Proc. of the VLDB Endowment, 2014;

\noindent
 [GM15] 	T. Gogacz, J. Marcinkowski,{\em 
The Hunt for a Red Spider: Conjunctive Query Determinacy Is Undecidable}; LICS 2015: 281-292;

 \noindent
 [GM16] 	T. Gogacz, J. Marcinkowski, {\em
Red Spider Meets a Rainworm: Conjunctive Query Finite Determinacy is Undecidable}; PODS 2016: 121-134;
 

%
\noindent
[JV09] V. Juge and M. Vardi,  {\em On the containment of Datalog
in Regular Datalog}; Technical report, Rice University,
2009;

%
%
%
\noindent
[LY85]
 Per-Ake Larson and H. Z. Yang,  {\em Computing queries from derived relations};
Proc. of the 11th International Conference on Very Large Data Bases - Volume
11, VLDB'85, pages 259--269. VLDB Endowment, 1985;
%
%
%

%
\noindent
[NSV06]
 A. Nash, L. Segoufin, and V. Vianu, {\em Determinacy and
rewriting of conjunctive queries using views: A progress report};
 Proc. of
ICDT 2007, LNCS vol. 4353;
pp 59--73;

\noindent
[NSV10]
 A. Nash, L. Segoufin, and V. Vianu.
 {\em Views and
queries: Determinacy and rewriting}; ACM Trans. Database Syst.,
35:21:1--21:41, July 2010;

\noindent
[P11]
 D. Pasaila, {\em Conjunctive queries determinacy and rewriting};  Proc. of the 14th  ICDT, pp. 220--231, 2011;

\noindent
[RRV15] J. Reutter, M. Romero,  M. Vardi, {\em Regular queries
on graph databases}; Proc. of the 18th ICDT; pp 177--194; 2015;

 %
\noindent
[V16] M.Y. Vardi, {\em A Theory of Regular Queries};
 PODS/SIGMOD keynote talk; Proc. of the 35th ACM  PODS 2016,
pp 1-9; 
%


\end{document}